\documentclass[conference]{IEEEtran}
\usepackage{xcolor}
\usepackage{xspace}
\usepackage{amsmath}
\usepackage{amsthm}
\usepackage{amssymb}
\usepackage{enumerate}
\usepackage[shortlabels]{enumitem}
\usepackage{thm-restate}
\usepackage{tikz}

\theoremstyle{definition}
\newtheorem{theorem}{Theorem}[section]
\newtheorem{definition}[theorem]{Definition}
\newtheorem{remark}[theorem]{Remark}
\newtheorem{proposition}[theorem]{Proposition}

\newtheorem{lemma}[theorem]{Lemma}
\newtheorem{corollary}[theorem]{Corollary}

\newcommand{\FO}{\textrm{\upshape FO}\xspace}
\newcommand{\FOxc}{\textrm{\upshape FO$_{\exists,\land}$}\xspace}

\newcommand{\GFO}{\textrm{\upshape GFO}\xspace}
\newcommand{\GNFO}{\textrm{\upshape GNFO}\xspace}
\newcommand{\UNFO}{\textrm{\upshape UNFO}\xspace}
\newcommand{\FOtwo}{\textrm{\upshape FO$^2$}\xspace}
\newcommand{\Ctwo}{\textrm{\upshape C$^2$}\xspace}
\newcommand{\GFOtwo}{\textrm{\upshape GFO$^2$}\xspace}

\begin{document}
\title{Craig Interpolation for Guarded Fragments}
\author{\IEEEauthorblockN{Balder ten Cate}
\IEEEauthorblockA{
ILLC, University of Amsterdam
\vspace{-5mm}
}
\and
\IEEEauthorblockN{Jesse Comer}
\IEEEauthorblockA{
ILLC, University of Amsterdam
\vspace{-5mm}
}}

\maketitle

\begin{abstract}
  We show that the guarded-negation fragment (\GNFO) is, in a precise sense,  the smallest extension
  of the guarded fragment (\GFO) with Craig interpolation. In contrast, 
  the smallest extension of the two-variable fragment (\FOtwo)
  with Craig interpolation is full first-order logic.
\end{abstract}

\section{Background}

\noindent\textbf{Decidable Fragments of \FO ~~}
The study of decidable fragments of first-order logic (\FO) is a 
topic with a long history (cf.~\cite{Borger1997:classic}). Inspired by Vardi~\cite{Vardi1996:why}, who asked ``what makes modal logic so robustly decidable?'' and Andreka et al.~\cite{Andreka1998:Modal}, who asked ``what makes modal logic tick?'' many decidable fragments have been introduced and studied over the last 25 years
that take
 inspiration from modal logic, which itself can be viewed as a fragment of \FO that features a restricted form of quantification. These include the following fragments, each of which naturally generalizes modal logic in a different 
way: the \emph{two-variable fragment} (\FOtwo)~\cite{Mortimer1975:languages},
the \emph{guarded fragment} (\GFO)~\cite{Andreka1998:Modal}, and the 
\emph{unary negation fragment} (\UNFO)~\cite{tencate2013:unary}.
Further decidable extensions of these fragments were subsequently identified, including the \emph{two-variable fragment with counting quantifiers} (\Ctwo)~\cite{Graedel97:two} and the 
\emph{guarded negation fragment} (\GNFO)~\cite{Barany2015:guarded}. The latter can be viewed as a
common generalization of \GFO and \UNFO. 
Many decidable logics used in computer science and AI, 
including various description logics and rule-based languages, 
can be translated into \GNFO and/or \Ctwo. In this sense, 
\GNFO and \Ctwo are convenient tools for explaining the decidability of other logics. 
Further extensions of \GNFO have been studied that push the decidability frontier even further (for instance with fixed-point operators and using clique-guards), as well as other, orthogonal decidable fragments. These fall outside the scope of this paper. 
Figure~\ref{fig:fragments} summarizes the fragments 
that are relevant for us here.
\looseness=-1

\medskip\par\noindent\textbf{The Craig Interpolation Property (CIP) ~~}
Ideally, an \FO-fragment is not only algorithmically but also
 model-theoretically well behaved. A particularly important
model-theoretic property of logics is the 
\emph{Craig Interpolation Property} (CIP). It states that, 
for all formulas $\varphi, \psi$, if 
 $\varphi \models \psi$, then there exists a formula $\vartheta$ such that $\varphi \models \vartheta$ and $\vartheta \models \psi$,
 and such that all non-logical symbols occurring in $\vartheta$
 occur both in $\varphi$ and in $\psi$.
 Craig~\cite{Craig1957}
 proved in 1957 that \FO itself has this property (hence the name). Several refinements of Craig's result 
 have subsequently been
 obtained (e.g.,~\cite{Otto2000:interpolation,Benedikt16:generating}). These have
 found numerous applications (e.g.,~\cite{tencate2013:beth,benedikt2016:query,jung2021:separating}). 
  While we have described CIP here as a 
 model theoretic property, it also has a proof-theoretic 
 interpretation. Indeed, it has been argued that CIP is an indicator for the existence of nice proof systems~\cite{hooglandthesis}.

 Turning our attention to the decidable fragments of \FO
 we mentioned earlier, it turns out that, although \GFO is
 in many ways
 model-theoretically well-behaved~\cite{Andreka1998:Modal},
 it lacks CIP~\cite{Hoogland02:interpolation}.
 Likewise, \FOtwo lacks CIP~\cite{Comer1969:classes}, although 
 the intersection of \GFO and \FOtwo (known as \GFOtwo) has CIP \cite{Hoogland02:interpolation}. 
 \Ctwo lacks
CIP as well (\cite[Example 2]{Jung2021:living} yields a counterexample).
 On the other hand, \UNFO and \GNFO have CIP
 \cite{tencate2013:unary,Benedikt2013:rewriting} 
 (as do their fixed-point extensions~\cite{Benedikt2015:interpolation,Benedikt2019:definability}). 
 Indeed, in the
 case of \UNFO and \GNFO, interpolants can be constructed effectively and tight bounds have been established on the 
 size of interpolants and the computational complexity of 
 computing them~\cite{Benedikt2015:effective}.
Figure~\ref{fig:fragments} summarizes these known results. 
Note: we restrict attention to relational signatures without constant symbols and function symbols.
Some of the results depend on this restriction.

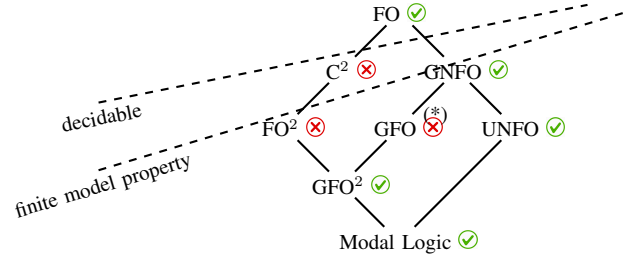
\begin{figure}\centering
\vspace{-3mm}
\newcommand{\yes}{\includegraphics[scale=.02]{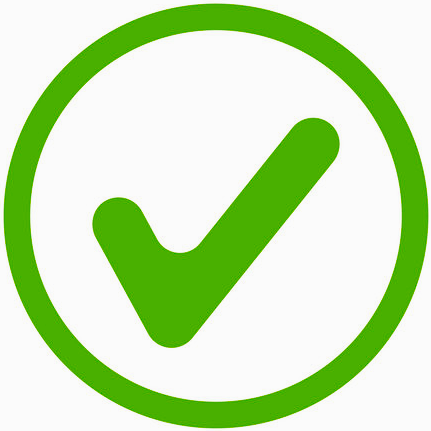}}
\newcommand{\no}{\includegraphics[scale=.02]{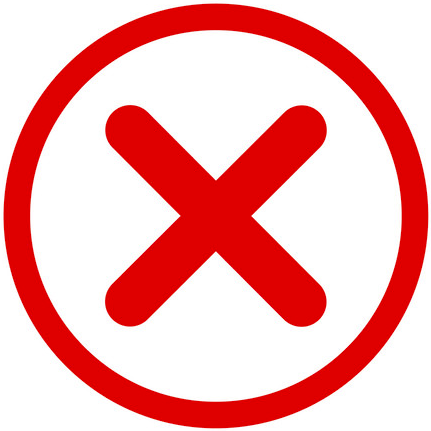}}
\begin{tikzpicture}
  [align=center,node distance=1.35cm, every node/.style={scale=0.8}]
    \node (fo) at (0,0) {\FO \yes};
    \node [below left of=fo] (ctwo) {\Ctwo \no\!\!\!\!\!};
    \node [below right of=fo] (gnfo) {\GNFO \yes\!\!\!\!\!};
    \node [below left  of=gnfo] (gfo)  {\GFO \no\!\!\!\!\!};
    \node [below right of=gnfo] (unfo) {\UNFO \yes\!\!\!\!\!};
    \node [below left of=ctwo] (twovar) {\FOtwo \no\!\!\!\!\!}; 
    \node [below left of= gfo] (twovargf) {\GFOtwo \yes\!\!\!\!\!};
    \node [below right of= twovargf] (ml) {Modal Logic \yes\!\!\!\!\!};

    \draw [thick, shorten <=-2pt, shorten >=-2pt] (fo) -- (gnfo);
    \draw [thick, shorten <=-2pt, shorten >=-2pt] (fo) -- (ctwo);
    \draw [thick, shorten <=-2pt, shorten >=-2pt] (ctwo) -- (twovar);
    \draw [thick, shorten <=-2pt, shorten >=-2pt] (gnfo) -- node[below] {~~(*)} (gfo);
    \draw [thick, shorten <=-2pt, shorten >=-2pt] (gnfo) -- (unfo);
    \draw [thick, shorten <=-2pt, shorten >=-2pt] (gfo) -- (twovargf);
    \draw [thick, shorten <=-2pt, shorten >=-2pt] (twovar) -- (twovargf);
    \draw [thick, shorten <=-2pt, shorten >=-2pt] (twovargf) -- (ml);
    \draw [thick, shorten <=-2pt, shorten >=-2pt] (unfo) -- (ml);
    
    \draw [dashed, thick] (-4,-1.2) -- 
      node[below, at start, sloped] {decidable} (2.5,0.1);
    \draw [dashed, thick] (-4,-2.1) -- node[below, at start, sloped] {finite model property} (3,0);
\end{tikzpicture}
\vspace{-2mm}
\caption{Some decidable fragments of \FO with (\yes) and without (\no) CIP. \\ The inclusion marked $(*)$ holds only for sentences and self-guarded formulas.}
\vspace{-5mm}
\label{fig:fragments}
\end{figure}

\medskip\par\noindent\textbf{What To Do When CIP Fails? ~}
When a fragment $L$ lacks CIP, the
question naturally arises as to whether there exists
a more expressive fragment $L'$ that
has CIP. If such $L'$ exists, then, 
in particular, interpolants for valid $L$-implications
can be found in $L'$. 
This line of analysis is sometimes referred to as 
``\emph{Repairing Interpolation}'' \cite{Areces03:repairing}. 
We will pursue this ``Repairing'' approach in the next sections. Before we do so, let us mention some other approaches
for dealing with fragments that lack CIP. 
One approach is to weaken CIP. 
For example, it was shown in \cite{Hoogland02:interpolation} that \GFO satisfies
a weak, ``modal'' form of Craig interpolation, where, roughly speaking, only
the relation symbols that occur in non-guard position in the interpolant are required to occur both in the premise and the conclusion. As it turns out, this weakening of CIP is strong enough to entail the (non-projective) \emph{Beth Definability Property},
which is one important use case of CIP.
Another recent approach~\cite{Jung2021:living} is to 
develop algorithms for testing whether an
interpolant exists for a given entailment. 
That is, rather than viewing Craig interpolation as a property of logics, 
the existence of interpolants is studied as an 
algorithmic problem at the level of individual entailments.
The interpolant existence problem turns out to be indeed
decidable (although of higher complexity than the satisfiability
problem) for both \GFO and \FOtwo~\cite{Jung2021:living}.
\looseness=-1

\section{Repairing Interpolation for \FOtwo}


The two-variable fragment (\FOtwo) consists of all
FO-formulas containing only two variables,
say, $x$ and $y$, where we allow for nested quantifiers that reuse the same variable (as in 
$\exists xy (R(x,y)\land \exists x(R(y,x)))$,
expressing the existence of a path of length 2). In this context, as is customary, we restrict attention to relations of arity at most $2$.
It is known that \FOtwo is decidable~\cite{Mortimer1975:languages} but does not have CIP~\cite{Comer1969:classes}. 

\begin{restatable}{theorem}{fotwothm}
\label{thm:FOtwo-main}
Let $L$ be any \FO-fragment that extends \FOtwo,
is closed under substitution, and has CIP. Then $\FO \preceq_{sent} L$.
\end{restatable}

Here, we write $L_1 \preceq_{sent} L_2$ to indicate that every $L_1$-sentence is expressible in $L_2$; by \emph{closure under substitution} we mean that,
for every formula $\varphi\in L$ containing an $n$-ary 
relation symbol $R$ and for every formula $\psi(x_1, \ldots, x_n)\in L$, $\varphi[\psi/R]$ is expressible 
in $L$, where $\varphi[\psi/R]$ is obtained from $\varphi$
by replacing every subformula of the form $R(y_1, \ldots y_n)$ by $\psi(y_1, \ldots, y_n)$ (assuming this is a safe substitution).
Intuitively, a fragment is closed under substitution  if it has a compositional syntax, and the assumption of closure under
substitution in Theorem~\ref{thm:FOtwo-main} serves the purpose of ensuring that
$L$ not only subsumes $\FOtwo$ but is also closed
under the connectives of $\FOtwo$ (indeed, this is all we use in the proof).

Theorem~\ref{thm:FOtwo-main} shows that, to repair interpolation for \FOtwo, we must go to full \FO. In particular, every extension of \FOtwo 
closed under substitution with CIP is (assuming an effective syntax)  undecidable.
The proof is given in Appendix~\ref{app:FOtwo}. 


 

\section{Repairing Interpolation for \GFO}



The guarded fragment (\GFO) allows formulas in
which all quantifiers are ``guarded''. Formally, 
a \emph{guard} for a formula $\varphi$ is an atomic
formula $\alpha$ whose free variables include all
free variables of $\varphi$. 
Following~\cite{Graedel99:restraining}, we allow  
$\alpha$ to be an equality. 
More generally, by an
\emph{$\exists$-guard} for $\varphi$, we will mean
a possibly-existentially-quantified atomic formula $\exists\overline{x}\beta$
whose free variables include all free variables of 
$\varphi$.
The formulas of \GFO are generated by the following grammar:
$$\varphi := \top \mid R(\overline{x}) \mid x = y \mid \lnot \varphi \mid \varphi \land \psi \mid \exists \overline{x} (\alpha \land \varphi),$$
where, in the last clause, $\alpha$ is a guard for $\varphi$. Note again that 
we do not allow
constants and function symbols.

In the guarded-negation fragment (\GNFO), arbitrary existential quantification is allowed, but every negation is required to be guarded. 
More precisely, the formulas of \GNFO are generated by the following grammar:
$$\varphi := \top \mid R(\overline{x}) \mid x = y \mid \varphi \lor \varphi \mid \varphi \land \varphi  \mid \exists x \varphi \mid \alpha \land \lnot \varphi,$$
where, in the last clause, $\alpha$ is a guard for $\varphi$.

As is customary, the above definitions are phrased
in terms of ordinary guards $\alpha$. However, it is easy to
see that if we allow for $\exists$-guards, this would not 
affect the expressive power (or the computational complexity)
of these logics in any way. This is because $\exists\overline{x}\beta\land \varphi$ can be equivalently written as $\exists\overline{x}(\beta\land\varphi)$. In other words, 
an $\exists$-guard is as good as an ordinary guard.

We call a \FO-formula \emph{self-guarded} if it
is either a sentence or it is of the form $\alpha\land\varphi$ where $\alpha$ is an $\exists$-guard for
$\varphi$. It was shown in~\cite{Barany2015:guarded} that
every self-guarded \GFO-formula is expressible in \GNFO.
In particular, this applies to all \GFO-sentences and
\GFO-formulas with at most one free variable 
(since any such formula can be equivalently written 
as $x=x\land\varphi$). 
It is therefore common to treat 
\GNFO as an extension of \GFO.
This is reflected
by the line marked (*) in Figure~\ref{fig:fragments}.
Formally, we write $L_1 \preceq_{sg} L_2$ to indicate
that every self-guarded $L_1$-formula is expressible in $L_2$;
hence $\GFO \preceq_{sg} \GNFO$.
\looseness=-1

Guarded fragments are peculiar, in that they are not closed
under substitution. For example $\exists xy(R(x,y)\land \neg S(x,y))$ belongs to \GFO but if we substitute $x=x\land y=y$ for $R(x,y)$, 
we obtain $\exists xy(x=x\land y=y\land\neg S(x,y))$, which
does not belong to \GFO (and is not even expressible in \GNFO).
\GFO and \GNFO are, however, closed under \emph{self-guarded substitution}: we can uniformly replace relations by
self-guarded formulas. 
\looseness=-1

Given these subtleties, we can now state our main result:
\begin{restatable}{theorem}{thmmain}
\label{thm:main}
Let $L$ be any FO-fragment such that 
\begin{enumerate}
    \item $\GFO\preceq_{sg} L$, 
    \item $L$ is closed under self-guarded substitution, 
    \item $L$ is closed under conjunction and disjunction, and
    \item $L$ has CIP.
\end{enumerate}
Then $\GNFO\preceq L$.
\end{restatable}

By $\GNFO\preceq L$, we mean that every $\GNFO$-formula
is equivalent to an $L$ formula (not only sentences, and
not only self-guarded formulas).

In other words, loosely speaking, \GNFO is the smallest extension of \GFO with CIP. 
The proof of Theorem~\ref{thm:main} is given in Appendix~\ref{app:GNF}. It is based on similar ideas as the proof of Theorem~\ref{thm:FOtwo-main}, but the argument is more intricate. 
The main technical result is the following
proposition:

\begin{restatable}{proposition}{propCQs}
\label{prop:CQs}
Let $L$ be any \FO-fragment with CIP that includes all atomic formulas and is closed under guarded quantification, conjunction, and unary implication. Then $\FOxc \preceq L$.
\end{restatable}

Here \FOxc denotes the existential-conjunctive
fragment of \FO (cf.~the Appendix); we say that a fragment $L$ is \textit{closed under guarded quantification} if, whenever  $\varphi\in L$ and $\alpha$ is a guard for $\varphi$, $L$ can express $\exists \overline{x} (\alpha \land \varphi)$; and
$L$ \emph{is closed under unary implications} if,
whenever $\varphi\in L$ and $\alpha$ is an atomic formula with only one free variable, $L$ can express $\alpha \to \varphi$.

We note that, for ease of exposition, Theorem~\ref{thm:FOtwo-main} and Theorem~\ref{thm:main} are stated 
in terms of fragments of \FO. However, the assumption that $L$ is a fragment of \FO is not used in any essential way in the proof. It is also possible to state these results using an abstract notion of logics, as in~\cite{tencate2005:interpolation,vanbenthem2009:lindstrom}.
It was shown in~\cite{tencate2005:interpolation} that every abstract logic extending \GFO with CIP is undecidable. However, \cite{tencate2005:interpolation} assumes constant symbols and concerns a stronger version of CIP, interpolating not only over relation symbols but also
over constant symbols.

\medskip\par\noindent\textbf{Acknowledgements:} We thank Jean  Jung, Frank Wolter, and Malvin Gattinger for  feedback on a  draft.
Balder ten Cate is supported by EU Horizon
2020 grant MSCA-101031081.

\bibliographystyle{plain}
\bibliography{main}

\bigskip\bigskip

\newcommand{\free}{\textrm{free}}
\newcommand{\bind}{\textsf{BIND}}

\appendices

\section{Proof of Theorem~\ref{thm:FOtwo-main}}
\label{app:FOtwo}

Several of the following proofs make use of second-order logic, with 
quantifiers over predicates. These second-order quantifiers will
be taken to range over the full powerset of the 
domain of the structure.

\fotwothm* 

\begin{proof}
The following proof uses a similar strategy as was used in~\cite{tencate2005:interpolation} to show that
every abstract modal language extending the modal language 
with the difference operator has full first-order
expressive power.

We will show by formula induction that, for every $\FO$-formula $\phi(x_1 \ldots, x_n)$ there is a sentence $\psi\in L$ over an extended signature containing additional
unary predicates $P_1, \ldots, P_n$, that is equivalent to 
$$\exists x_1\ldots x_n(\big(\!\!\!\bigwedge_{i=1\ldots n}\!\!\! P_i(x_i)\land\forall y(P_i(y)\to y=x_i)\big)\land\phi(x_1, \ldots, x_n)).$$ 

In other words, $\psi$ is a sentence expressing that $\phi$ holds under an assignment of its free variables to some tuple of elements which uniquely satisfy the $P_i$ predicates. In the case that $n=0$ (i.e., the case that $\phi$ is a sentence), we then have that $\psi$ is equivalent to $\phi$, which show that $FO\preceq_{sent} L$.

The base case of the induction is straightforward (recall that we restrict attention to relations of arity at most 2). The induction step for the Boolean connectives is straightforward as well (using closure under substitution). In fact, 
the only non-trivial part of the argument is the induction step for the existential quantifier. 
Let $\phi(x_1,\ldots,x_n)$ be of the form $\exists x_{n+1} \phi'(x_1.\ldots,x_n,x_{n+1})$. 
By induction, there is an $L$-sentence $\psi$
over the signature with additional unary predicates $P_1, \ldots, P_{n+1}$, corresponding to $\phi'(x_1, \ldots, x_n, x_{n+1})$.
Now, let $\psi'$ be obtained from $\psi$ by replacing every occurrence of $P_{n+1}$ by $P'$ for some fresh unary predicate $P'$. Furthermore, let
\[ \begin{array}{ll} 
\gamma(x) &:= \psi \land P_{n+1}(x),\\[2mm]
\chi(x)   &:= (P'(x)\land\forall y(P'(y)\to y=x)) \to \psi'.
\end{array} \]
(where $x$ is either of the two variables we have at our disposal; it does not matter which). It follows from closure under substitution that both can be written as an $L$-formula.
Then $$\gamma(x) \models \chi(x).$$
Let $\theta(x)\in L$ be an interpolant. 
By closure under substitution, $\exists x \theta(x)$ is expressible in $L$ as well.
We claim that this sentence satisfies the requirement of our 
claim.

To see this, first observe that since $P_{n+1}$  occurs only in the antecedent and $P'$ only in the consequent, the following second-order entailment is also valid:
$$\exists P_{n+1}\gamma(x) \models \vartheta(x) \models \forall P' \chi(x).$$

It is not hard to see that $\exists P_{n+1} \gamma(x)$ and $\forall P' \chi(x)$ are equivalent. Indeed, both are
satisfied in a structure $M$ under an assignment $g$ precisely
if $M',g\models\phi$, where $M'$ is the expansion of $M$ in 
which $P_{n+1}$ denotes the singleton set $\{g(x_{n+1})\}$.

It then follows that $\vartheta(x)$, being sandwiched between the two, is also equivalent to $\exists P_{n+1} \gamma(x)$. Therefore, $\exists x \vartheta(x)$ is 
equivalent to $\exists x \exists P_{n+1} \gamma(x)$, which 
is equivalent to $\exists P_{n+1} \psi$,
which clearly satisfies the requirement of our claim.
\end{proof}

This further implies undecidability, under a mild extra condition: 
we say that a fragment $L$ of $\FO$ is ``effectively closed under conjunction'', if there is a computable function
that takes any two formulas $\phi,\psi\in L$ and 
outputs a formula $\chi\in L$ such that $\chi$ is 
equivalent to $\phi\land\psi$.

\begin{corollary}
Let $L$ be any \FO-fragment that extends \FOtwo,
is closed under substitution, and has CIP. 
Furthermore, assume that $L$ is effectively 
closed under conjunction. Then the satisfiability
problem for $L$ is undecidable. 
\end{corollary}

\begin{proof}
It is known that satisfiability is undecidable
for $FO^2$-formulas with two transitive relations~\cite{Kieronski2005Results}.
This problem reduces to the satisfiability
problem for $L$ as follows: let $\phi$
be any $FO^2$-formula containing (among possibly
other relation symbols) binary relation symbols
$R_1$ and $R_2$. Then $\phi$ is satisfiable over
structures in which $R_1$ and $R_2$ are 
transitive, if and only if 
$\phi\land\psi$ is satisfiable, where $\psi$
is a (fixed) $L$-sentence expressing that $R_1$ 
and $R_2$ are transitive. Note that it follows from Theorem~\ref{thm:FOtwo-main} that such a 
sentence $\psi$ exists. Since $L$ is effectively
closed under conjunction, this is an effective 
reduction.
\end{proof}

\section{Proof of Theorem~\ref{prop:CQs}}
\label{app:GNF}

We will assume familiarity with conjunctive queries (CQs)
and unions of conjunctive queries (UCQs).
An important alternative characterization for \GNFO is that it is the logic which can express every union of conjunctive queries (UCQ) and is closed under guarded negation~\cite{Barany2015:guarded}. This is made explicit in the following equivalent grammar for \GNFO:
$$\varphi := R(\overline{x}) \mid x = y \mid \alpha \land \lnot \varphi \mid q[\varphi_1 / R_1, \hdots, \varphi_n / R_n],$$
where $q$ is a UCQ with relation symbols $R_1, \hdots, R_n$ and $\varphi_1, \hdots, \varphi_n$ are self-guarded formulas with the appropriate number of free variables and generated by the same recursive grammar. We refer to this as the UCQ syntax for \GNFO. 

The main thrust of the argument will be to show that our abstract logic $L$ can express all positive existential formulas, from which it will follow easily that $L$ is able to express all formulas in the UCQ syntax for \GNFO.

\begin{definition}
We write \FOxc for the fragment of first-order logic with only existential quantification and conjunction:
$$\varphi := R(x_1, \hdots, x_k) \mid x = y \mid \varphi \land \varphi \mid \exists x \varphi.$$
\end{definition}

\begin{definition}
Let $\varphi$ be a formula in \FOxc, let $\overline{y} = y_1, \hdots, y_n$ be a tuple of distinct variables,
and let $\overline{P}=P_1, \hdots, P_n$ be a tuple of unary predicates  of the same length. Then $\bind_{\overline{y}\mapsto \overline{P}}(\varphi)$ is defined recursively as follows:
\[\begin{array}{ll}
\bind_{\overline{y}\mapsto \overline{P}}(\alpha) &= 
\exists \overline{y}' (\alpha \land \bigwedge_{1\leq i\leq n, y_i\in \free(\alpha)}P_i(u_i)) \\
\bind_{\overline{y}\mapsto \overline{P}}(\phi \land \psi) &= \bind_{\overline{y}\mapsto \overline{P}}(\phi) \land \bind_{\overline{y}\mapsto \overline{P}}(\psi) \\
\bind_{\overline{y}\mapsto \overline{P}}(\exists z \psi) &= \exists z (\bind_{\overline{y}\mapsto \overline{P}}(\psi)),
\end{array}\]
where $\alpha$ is an atomic fact (possibly an equality),
and $\overline{y}'$ is the  restriction
of $\overline{y}$ to variables occurring in $\alpha$.
If no variable in $\overline{y}$ occurs in $\alpha$, 
$\bind_{\overline{y}\mapsto \overline{P}}(\alpha)$ is
understood to be
simply $\alpha$.
\end{definition}

\begin{remark}
The free variables of $\bind_{\overline{y} \mapsto \overline{P}}(\varphi)$, for $\overline{y}=y_1,\hdots,y_n$, are exactly $\free(\varphi)\setminus\{y_1, \hdots, y_n\}$. This justifies our use of the word ``BIND''.
\end{remark}


\begin{proposition}
\label{BINDcomp}
For all $\FOxc$-formulas $\varphi$ and for all 
$\overline{x}, \overline{y}$ and $\overline{P},\overline{Q}$, if $\overline{x}$ and $\overline{y}$
are disjoint, then
$$\bind_{\overline{x}\overline{y} \mapsto \overline{P}\overline{Q}}(\varphi) \equiv \bind_{\overline{x} \mapsto \overline{P}}(\bind_{\overline{y} \mapsto \overline{Q}}(\varphi)).$$
\end{proposition}

\begin{definition}
We call a formula $\varphi$ \textit{clean} if no free variable of $\varphi$ also occurs bound in $\varphi$, and $\varphi$ does not contain two quantifiers for the same
variable.
\end{definition}

\begin{proposition}
\label{impLemma}
For every clean \FOxc-formula $\varphi$, for every tuple of distinct variables $\overline{y}=y_1, \hdots, y_n$ (with each $y_i\in \free(\varphi)$), and for every tuple of unary predicates
$\overline{P}=P_1,\ldots,P_n$, we have that
$$\big(\bigwedge_{i=1\ldots n} P_i(y_i)\big)\models \varphi  \to \bind_{\overline{y} \mapsto \overline{P}}(\varphi).$$
\end{proposition}

\begin{proof} The proof is by induction on $\varphi$. More precisely,
the induction hypothesis states that, for every
model $M$ and variable assignment $g$, if 
$M,g\models \bigwedge_{i=1\ldots n} P_i(y_i)$
and 
 $M,g\models\varphi$ then $M,g\models \bind_{\overline{y} \mapsto \overline{P}}(\varphi)$.
\end{proof}

\begin{proposition}
\label{eqLemma}
For every clean \FOxc-formula $\varphi(x,\overline{y})$ with $\overline{y}=y_1, \ldots, y_n$ distinct from $x$, and for every $n$-tuple of unary predicates
$\overline{P}=P_1, \ldots, P_n$ not occurring in $\varphi$, we have that
$$\exists x\varphi(x,\overline{y}) \equiv \forall \overline{P}\Big(\big(\bigwedge_{i=1\ldots n} P_i(y_i)\big) \to \exists x \bind_{\overline{y} \mapsto \overline{P}}(\varphi(x,\overline{y}))\Big).$$
\end{proposition}

\begin{proof} \hspace*{1pt} \\
The left-to-right entailment follows from Proposition \ref{impLemma}: suppose $M,g\models \exists x\varphi(x,\overline{y})\land\bigwedge_{i=1\ldots n} P_i(y_i)$. Then $M,g[x/b]\models \varphi(x,\overline{y})\land \bigwedge_{i=1\ldots n} P_i(y_i)$ for some $b\in M$. 
Then, by Proposition~\ref{impLemma}, 
$M,g[x/b]\models \bind_{\overline{y}\mapsto\overline{P}}(\varphi(x,\overline{y}))$, and hence
$M,g\models \exists x\bind_{\overline{y}\mapsto\overline{P}}(\varphi(x,\overline{y}))$.

For the reverse direction, suppose $M,g \models \forall \overline{P}(\bigwedge_i P_i(y_i) \to \exists x \bind_{\overline{y} \mapsto \overline{P}}(\varphi(x,\overline{y})))$. Let $M'$ be the expansion
of the structure $M$ in which each unary predicate symbol $P_i$ is interpreted as $\{g(y_i)\}$. Then, 
by the semantics of second-order quantifiers, 
we have that $M',g\models \exists x \bind_{\overline{y} \mapsto \overline{P}}(\varphi(x,\overline{y}))$, and hence $M',g[x/b] \models \bind_{\overline{y} \mapsto \overline{P}}(\varphi(x,\overline{y}))$ for some $b \in M$. 
To complete the proof, it suffices to show that
$M',g[x/b]\models \varphi(x,\overline{y})$ (since 
this implies that also $M,g[x/b]\models\varphi(x,\overline{y})$).

For any subformula containing a bound occurrence of a variable $y_i \in \overline{y}$, we have that any witness for that variable $y_i$ must also be in $P_i$ (by construction of $\bind_{\overline{y} \mapsto \overline{P}}(\varphi(x,\overline{y}))$ and the assumption that $\varphi(x,\overline{y})$ is clean). Since each $P_i$ is a singleton, this implies that each witness for $y_i$ in any subformula is $g(y_i)$. It follows that $M, g[\overline{y}'/\overline{a}'] \models \alpha$ for each atomic formula $\alpha$ occurring in $\varphi(x,\overline{y})$, where $\overline{y}'$ is the tuple of variables of $\overline{y}$ occurring in $\alpha$. By a simple subformula induction, we then obtain that $M \models \varphi(b,\overline{a})$, completing the proof.
\end{proof}

\begin{lemma}
\label{BINDexpLemma}
Let $L$ be any \FO-fragment which can express atomic facts and is closed under guarded quantification, conjunction, and unary implication. If $L$ can express $\varphi \in \FOxc$ and all of its subformulas, then $L$ can express $\bind_{\overline{y} \mapsto \overline{P}}(\varphi)$.
\end{lemma}

\begin{proof} \hspace*{1pt} \\
We show by strong induction on the complexity of the \FOxc-formula $\varphi$ that this proposition holds.

\hspace*{1pt} \\
\textbf{Base Case} \\
If $\varphi$ is an atomic fact and $\overline{y}=y_1 \ldots, y_n$, then \[\bind_{\overline{y} \mapsto \overline{P}}(\varphi) = \exists \overline{y} (\varphi \land \bigwedge_{1\leq i\leq n, y_i\in \free(\alpha)} P_i(y_i)),\] 
which $L$ can express by closure under conjunction and guarded quantification.

\hspace*{1pt} \\
\textbf{Inductive Step} \\
Suppose that $\varphi = \psi_1 \land \psi_2$. Since $L$ can express $\varphi$ and all of its subformulas, it can also express $\psi_1$, $\psi_2$, and all of their subformulas. Then by the inductive hypothesis, $L$ can express $\bind_{\overline{y} \mapsto \overline{P}}(\psi_1)$ and $\bind_{\overline{y} \mapsto \overline{P}}(\psi_2)$. Then by closure under conjunctions, $L$ can express $\bind_{\overline{y} \mapsto \overline{P}}(\varphi) = \bind_{\overline{y} \mapsto \overline{P}}(\psi_1) \land \bind_{\overline{y} \mapsto \overline{P}}(\psi_2)$.

\hspace*{1pt} \\
Next, suppose that $\varphi(\overline{x},\overline{y}) = \exists z \psi(\overline{x},\overline{y},z)$. We need to show that $L$ can express $\bind_{\overline{y} \mapsto \overline{P}}(\varphi(\overline{x},\overline{y}))$, which, by definition, is the same as $\exists z (\bind_{\overline{y} \mapsto \overline{P}}(\psi(\overline{x},\overline{y},z)))$.

Since $L$ can express $\varphi$ and all of its subformulas, it can also express $\psi$ and all of its subformulas. Then, by the inductive hypothesis, 
$L$ can express $\bind_{\overline{y} \mapsto \overline{P}}(\psi)$ as well as
$\bind_{\overline{x}\overline{y} \mapsto \overline{Q}\overline{P}}(\psi)$. By closure under conjunction and guarded quantification, it follows that $L$ can express
\[\gamma(\overline{x}) :=\exists z (G(\overline{x},z) \land \bind_{\overline{y} \mapsto \overline{P}}(\psi))\] 
and
\[\exists z (z=z \land \bind_{\overline{x}\overline{y} \mapsto \overline{Q}\overline{P}}(\psi)),\]
where $G$ is a fresh relation symbol not occurring in $\psi$. Then by closure under unary implications, we have that $L$ can also express
$$\chi(\overline{x}) := \big(\bigwedge_{i} Q_i(x_i)\big) \to \exists z (z=z \land \bind_{\overline{x}\overline{y} \mapsto \overline{Q}\overline{P}}(\psi)).$$

\medskip\par\noindent\textbf{Claim: }
$\gamma(\overline{x})\models \chi(\overline{x})$

\medskip\par\noindent\emph{Proof of claim:}
By Proposition \ref{BINDcomp}, 
\begin{equation}
\bind_{\overline{x}\overline{y} \mapsto \overline{Q}\overline{P}}(\psi) \equiv \bind_{\overline{x} \mapsto \overline{Q}}(\bind_{\overline{y} \mapsto \overline{P}}(\psi))
\label{eq:bind-bind}
\end{equation} 
Therefore, by Proposition \ref{impLemma},
$$\bind_{\overline{y} \mapsto \overline{P}}(\psi) \models \big(\bigwedge_{i} Q_i(x_i)\big) \to  \bind_{\overline{x}\overline{y} \mapsto \overline{Q}\overline{P}}(\psi),$$
From this, it  follows that
$$\exists z(\bind_{\overline{y} \mapsto \overline{P}}(\psi)) \models \big(\bigwedge_{i} Q_i(x_i)\big) \to \exists z\bind_{\overline{x}\overline{y} \mapsto \overline{Q}\overline{P}}(\psi),$$
(because $z$ is distinct from $x_i$)
and therefore
$\gamma(\overline{x})\models \chi(\overline{x})$. This concludes the proof of the claim.

\medskip
Since $L$ can express both $\gamma(\overline{x})$ and $\chi(\overline{x})$, we have by the Craig interpolation property that $L$ can express some Craig interpolant $\vartheta(\overline{x})$. Since $G$ and the $Q_i$ predicates do not occur in $\varphi$, they do not occur in $\vartheta(\overline{x})$, and therefore, the following
second-order implication is valid:
$$\exists G \gamma(\overline{x}) \models \vartheta(\overline{x}) \models \forall P \chi(\overline{x}).$$

It is easy to see that $\exists G \gamma(\overline{x}) \equiv \exists z \bind_{\overline{y} \mapsto \overline{P}}(\psi)$. Similarly, it follows from Proposition \ref{eqLemma} and equation (\ref{eq:bind-bind}) that $\forall \overline{P} \chi(\overline{x}) \equiv \exists z \bind_{\overline{y} \mapsto \overline{P}}(\psi)))$. Hence 
$$\exists z \bind_{\overline{y} \mapsto \overline{P}}(\psi) \models \vartheta(\overline{x}) \models \exists z \bind_{\overline{y} \mapsto \overline{P}}(\psi)$$
Therefore, $\vartheta(\overline{x}) \equiv \exists z \bind_{\overline{y} \mapsto \overline{P}}(\psi)$. In particular, this means that $\exists z \bind_{\overline{y} \mapsto \overline{P}}(\psi)$ is expressible in $L$.
\end{proof}


We are now ready to prove Proposition~\ref{prop:CQs}, restated below.

\propCQs*

\begin{proof} \hspace*{1pt} \\
By strong induction on formulas $\varphi$ of \FOxc. The base case is immediate, since $L$ can express all atomic formulas. For the inductive step, if $\varphi := \psi_1 \land \psi_2$, then by the inductive hypothesis, $L$ can express $\psi_1$ and $\psi_2$, and so by closure under conjunction, $L$ can express $\varphi$. Now suppose $\varphi(\overline{y}) := \exists x(\psi(x,\overline{y}))$.
By the inductive hypothesis, together with closure
under guarded quantification, $L$ can express
\[\gamma(\overline{y}) := \exists x (G(x, \overline{y}) \land \psi).\]
Furthermore, 
by Lemma \ref{BINDexpLemma}, we have that $L$ can express $\bind_{\overline{y} \mapsto \overline{P}}(\psi)$, and therefore, by closure under guarded quantification and unary implications, $L$ can express \[\chi(\overline{y}) := \big(\bigwedge_{i} P_i(y_i)\big) \to \exists x(x=x\land \bind_{\overline{y} \mapsto \overline{P}}(\psi)).\] 

\medskip\par\noindent\textbf{Claim: } 
$\gamma(\overline{y}) \models \chi(\overline{y})$.

\medskip\emph{Proof of claim:}
It is clear that $\gamma(\overline{y})\models \exists x \psi$.
Furthermore, by Proposition \ref{impLemma},
$\psi\models\big(\bigwedge_{i} P_i(y_i)\big) \mapsto \bind_{\overline{y} \mapsto \overline{P}}(\psi)$, from which it follows that
$\exists x \psi \models\chi(\overline{y})$
(since the variable $x$ is distinct from $y_1, \ldots, y_n$).
Therefore, $\gamma(\overline{y}) \models \chi(\overline{y})$.
\medskip

Let $\vartheta(\overline{y})$ be any interpolant for $\gamma(\overline{y}) \models \chi(\overline{y})$ in $L$.
Since $G$ and the predicates in $\overline{P}$ do not occur in $\psi$, we then have that  the following second-order entailments are
valid:
$$\exists G \exists x (G(x, \overline{y}) \land \psi) \models \vartheta(\overline{y}) \models \forall \overline{P} ((\bigwedge_{i} P_i(y_i)) \to \exists x \bind_{\overline{y} \mapsto \overline{P}}(\psi)).$$
It is easy to see that
$$\exists G \exists x (G(x, \overline{y}) \land \psi) \equiv \exists x \psi.$$
Furthermore, by Lemma \ref{eqLemma},
$$\psi \equiv \forall \overline{P} ((\bigwedge_{i} P_i(y_i)) \to \bind_{\overline{y} \mapsto \overline{P}}(\psi)).$$
from which it follows that
$$\exists x \psi \equiv \forall \overline{P} ((\bigwedge_{i} P_i(y_i)) \to \exists x \bind_{\overline{y} \mapsto \overline{P}}(\psi))$$
(since $x$ is distinct from $y_1, \ldots, y_n$).

Therefore, $\vartheta(\overline{y}) \equiv \varphi(\overline{y})$, and so we are done.
\end{proof}

We are now ready to prove the main result.

\thmmain*

\begin{proof} \hspace*{1pt} \\
Since $L$ can express self-guarded \GFO-formulas, it can express formulas of the form $\exists \overline{x} \beta$, where $\beta$ is an atomic formula. Thus by closure under self-guarded substitution, we have that $L$ is closed under guarded quantification. Furthermore, $L$ can express any self-guarded formula of the form $\alpha \land \lnot \beta$, where $\alpha$ and $\beta$ are atomic formulas such that $free(\alpha) = free(\beta)$. Then for any formula $\varphi$ expressible in $L$ with $free(\varphi) \subseteq free(\beta)$, $\alpha \land \varphi$ is a self-guarded formula. Thus by self-guarded substitution, $L$ can also express $\alpha \land \lnot (\alpha \land \varphi)$, which is equivalent to $\alpha \land \lnot \varphi$; hence $L$ is closed under guarded negation. If $L$ can express $\varphi$, then by closure under guarded negation and disjunction, it can also express $(x=x \land \lnot P(x)) \lor \varphi$, which is equivalent to $P(x) \to \varphi$. Hence $L$ is closed under unary implications. Therefore, by Theorem \ref{prop:CQs}, $L$ can express all formulas in \FOxc. Then by expressibility of disjunction, $L$ can express all unions of conjunctive queries. The result then follows immediately from the UCQ-syntax for \GNFO, by closure under self-guarded substitution.
\end{proof}

\end{document}